\documentclass[runningheads]{llncs}
\usepackage{amssymb,amsfonts,amsmath}
\usepackage{cite}
\usepackage{graphicx}
\usepackage{array}
\usepackage{algorithm}
\usepackage{algorithmicx}
\usepackage[noend]{algpseudocode}
\usepackage[verbose]{wrapfig}
\usepackage[colorlinks,linkcolor=cyan,citecolor=blue]{hyperref}

\def\cA{{\cal A}}
\def\cL{{\cal L}}
\def\cT{{\cal T}}
\def\dfs{\mathrm{DFS}}
\def\add{\mathrm{add}}
\def\asq{\mathrm{Asquare}}
\def\rep{\mathrm{repeat}}
\def\last{\mathrm{last}}
\def\free{\mathrm{free}}
\def\gr{{\sf gr}}
\def\ch{{\sf ch}}
\def\re{{\sf re}}
\def\br{{\sf br}}
\def\Bf{{\sf bf}}
\def\SF{{\sf SF}}
\def\PF{{\sf PF}}
\def\ASF{{\sf ASF}}
\def\ACF{{\sf ACF}}
\def\A4F{{\sf A4F}}

\begin{document}

\title{Branching Frequency and Markov Entropy of Repetition-Free Languages}
\author{Elena A. Petrova \and Arseny M. Shur\thanks{Supported by the Ministry of Science and Higher Education of the Russian Federation (Ural Mathematical Center project No. 075-02-2020-1537/1).}}
\institute{Ural Federal University, Ekaterinburg, Russia\\ \{elena.petrova,arseny.shur\}@urfu.ru}
\maketitle

\begin{abstract}
We define a new quantitative measure for an arbitrary factorial language: the entropy of a random walk in the prefix tree associated with the language; we call it Markov entropy. We relate Markov entropy to the growth rate of the language and to the parameters of branching of its prefix tree. We show how to compute Markov entropy for a regular language. Finally, we develop a framework for experimental study of Markov entropy by modelling random walks and present the results of  experiments with power-free and Abelian-power-free languages.
\keywords{Power-free language, Abelian-power-free language,  Markov entropy, prefix tree, random walk}
\end{abstract}

\section{Introduction}
Formal languages closed under taking factors of their elements (factorial languages) are natural and popular objects in combinatorics. Factorial languages include sets of factors of infinite words, sets of words avoiding patterns or repetitions, sets of minimal terms in algebraic structures, sets of palindromic rich words and many other examples. One of the main combinatorial parameters of factorial languages is their asymptotic growth. Usually, ``asymptotic growth'' means asymptotic behaviour of the function $C_L(n)$ which returns the number of length-$n$ words in the language $L$. (In algebra, the function which counts words of length \emph{at most $n$} is more popular.)

In this paper we propose a different parameter of asymptotic growth, based on representation of factorial languages as \emph{prefix trees} which are diagrams of the prefix order on words. Given such an infinite directed tree, one can view each word as a walk starting at the root. If we consider \emph{random walks}, in which the next node is chosen uniformly at random among the children of the current node, we can define their \emph{entropy} (the measure of expected uncertainty of a single step). As a random walk is a Markov chain, we call this parameter the Markov entropy of a language. This parameter was earlier considered for a particular subclass of regular factorial languages in the context of antidictionary data compression \cite{CMRS99}. However, it seems that more general cases were not analysed up to now. Our interest to Markov entropy is twofold. First, it allows to estimate growth properties of a language from statistics of experiments where exact methods do not work. Second, it is related to a natural and efficient (at least theoretically) data compression scheme, which encodes the choices made during a walk in the prefix tree.

Our contribution is as follows. In Section~\ref{s:entropy} we define order-$n$ Markov entropy $\mu_n(L)$ of a language $L$ through length-$n$ random walks in its prefix tree $\cT(L)$ and the Markov entropy $\mu(L)=\lim_{n\to\infty} \mu_n(L)$. Then we relate Markov entropy to the exponential growth rate of $L$ and to the parameter called \emph{branching frequency} of a walk in $\cT(L)$. In Section~\ref{ss:reg} we show how to compute Markov entropy for a regular language. Then in Section~\ref{ss:walks} we propose a model of random walk for an arbitrary factorial language through depth-first search and show how to recover branching frequency from observable parameters of a walk. Finally, in Section~\ref{s:exper} we present algorithms used in the experimental study of Markov entropy for power-free and Abelian-power-free languages and the results of this study.

\section {Preliminaries}

We study words and languages over finite alphabets; $\Sigma^*$ denotes the set of all words over an alphabet $\Sigma{=}\{0,\ldots,\sigma{-}1\}$. Standard notions of prefix, suffix, factor are used. We use the array notation $w=w[1..n]$ for a word of length $n=|w|$; thus $w[i..i{+}k{-}1]$ stands for the length-$k$ factor of $w$ starting at position $i$. In particular, $w[i..i]=w[i]$ is the $i$th letter of $w$ and $w[i..i{-}1]$ is the empty word, denoted by $\lambda$. 
 A word $w$ is \emph{right extendable} in a language $L$ if $L$ contains infinitely many words with the prefix $w$; $\re(L)$ denotes the set of all words which are right extendable in $L$.

A word $w$ has \emph{period} $p$ if $w[1..|w|{-}p]=w[p{+}1..w]$. For an integer $k>1$, the \emph{$k$-power} of a word $w$ is the concatenation of $k$ copies of $w$. For an arbitrary real $\beta>1$, the \emph{$\beta$-power} (resp., the \emph{$\beta^+\!$-power}) of $w$ is the prefix of length $\lceil \beta|w|\rceil$ (resp., $\lfloor \beta|w|+1\rfloor$) of the infinite word $w^\infty=ww\cdots w\cdots$. E.g., $(010)^{2^+}=(010)^{7/3}=0100100$, $(010)^{5/2}=(010)^{(5/2)^+}=01001001$. A word is \emph{$\beta$-power-free} if it has no $\beta$-powers as factors; the $k$-ary \emph{$\beta$-power-free language} $\PF(k,\beta)$ consists of all $\beta$-power-free words over the $k$-letter alphabet. The same definitions apply to $\beta^+\!$-powers. The crucial result on the power-free languages is Threshold theorem, conjectured by Dejean \cite{Dej72} and proved by efforts of many authors \cite{Pan84c,Mou92,Car07,MoCu07,CuRa11,Rao11}. The theorem establishes the boundary between finite and infinite power-free languages: the minimal infinite $k$-ary power-free languages are $\PF(3,{\frac 74}^+)$, $\PF(4,{\frac 75}^+)$, and $\PF(k,{\frac k{k-1}}\!^+)$ for $k=2$ and $k\ge 5$. These languages are called \emph{threshold languages}.

\emph{Parikh vector} of a word $w$, denoted by $\vec V(w)$, is a length-$\sigma$ vector such that $\vec V(w)[i]$ is the number of occurrences of the letter $i$ in $w$ for each $i\in\Sigma$. Two words with equal Parikh vectors are said to be \emph{Abelian equivalent}. A concatenation of $k$ Abelian equivalent words is an \emph{Abelian $k$th power}. Abelian $k$-power-free words are defined similar to $k$-power-free words; Abelian square-free (resp., cube-free, 4-power-free) languages over four (resp., three, two) letters are infinite \cite{Dek79, Ker92}.

A language $L\subseteq\Sigma^*$ is \emph{factorial} if it contains all factors of each its element. Power-free and Abelian-power-free languages are obviously factorial. The relation ``to be a prefix (resp., a suffix, a factor)'' is a partial order on any language. The diagram of the prefix order of a factorial language $L$ is a directed tree $\cT(L)$ called \emph{prefix tree}\footnote{One can choose to study the tree obtained from the suffix order in a dual way, but if a language is closed under reversal, as in the case of power-free languages, then these two trees are isomorphic.}. Prefix trees are main objects of study in this paper. For convenience, we assume that an edge of the form $(w,wa)$ in $\cT(L)$ is labeled by the letter $a$; in this way, the path from the root to $w$ is labeled by $w$.

For regular languages we use deterministic finite automata with partial transition function (PDFA), viewing them as labelled digraphs. We assume that all states of a PDFA are reachable from the initial state; since we study factorial languages, we also assume that all states are final (so a PDFA accepts a word iff it can read it). When a PDFA $\cA$ is fixed, we write $q.w$ for the state of $\cA$ obtained by reading $w$ starting at the state $q$.

\emph{Combinatorial complexity} (or \emph{growth function}) of a language $L\subseteq\Sigma^*$ is a function counting length-$n$ words in $L$: $C_L(n)=|L\cap\Sigma^n|$. The \emph{growth rate} $\gr(L)=\limsup_{n\to\infty}(C_L(n))^{1/n}$ describes its asymptotic growth. Combinatorial complexity of factorial languages is submultiplicative: $C_L(m+n)\le C_L(m)C_L(n)$; by Fekete's lemma \cite{Fek23}, this implies $\gr(L)=\lim_{n\to\infty}(C_L(n))^{1/n}=\inf_{n\in\mathbb{N}}(C_L(n))^{1/n}$. A survey of techniques and results on computing growth rates for regular and power-free languages can be found in \cite{Sh12rev}.

\paragraph*{Infinite trees.} We consider infinite $k$-ary rooted trees: the number of children of any node is at most $k$. Nodes with more than one child are called \emph{branching points}. The \emph{level} $|u|$ of a node $u$ is the length of the path from the root to $u$. A subtree $\cT_u$ of a tree $\cT$ consists of the node $u$ and all its descendants. The tree $\cT$ is \emph{$p$-periodic} (resp., \emph{$p$-subperiodic}) if there exists a function $f$ on the set of nodes such that each subtree $\cT_u$ is an isomorphic copy (resp., is a subgraph) of the subtree $\cT_{f(u)}$ and $|f(u)|\le p$. The prefix tree of any factorial language $L$ is 0-subperiodic, since suffixes of elements of $L$ are also in $L$. Furthermore, $\cT(L)$ is $p$-periodic  for some $p$ iff $L$ is regular\footnote{Note that $p$-periodicity means exactly that $L$ has finitely many quotients, which is equivalent to regularity.}.

There are two widely used parameters of growth for infinite trees; see, e.g., \cite{LyPe16}. ``Horizontal'' growth is measured by the \emph{growth rate} $\gr(\cT)= \lim_{n\to\infty} (T_n)^{1/n}$, where $T_n$ is the number of nodes of level $n$, whenever this limit exists. Hence, $\gr(\cT(L))=\gr(L)$. ``Vertical'' growth is measured by the \emph{branching number} $\br(\cT)$, which is usually defined using the notion of network flow. However, the result of Furstenberg \cite{Fur67} says that $\br(\cT)=\gr(\cT)$ for subperiodic trees, so for prefix trees we have only one parameter. In Section~\ref{s:entropy}, we propose one more parameter of growth using the notion of entropy.

\paragraph*{Entropy.} Let $\xi=({x_1}_{|p_1},\ldots, {x_n}_{|p_n})$ be a discrete finite-range random variable, where $p_i, i=1,\ldots,n$, is the probability of the outcome $x_i$. The \emph{entropy} of $\xi$ is the average amount of information in the outcome of a single experiment: $H(\xi)=-\sum_{i=1}^k p_i\log p_i$ (throughout the paper, $\log$ stands for the binary logarithm). Lemma~\ref{l:shannon} below contains basic properties of entropy, established by Shannon \cite{Shan48}. For more details we refer the reader to the book \cite{Ash65}.

\begin{lemma} \label{l:shannon}
(1) For a random variable $\xi=({x_1}_{|p_1},\ldots, {x_n}_{|p_n})$, $H(\xi)\le \log n$; the equality holds for the uniform distribution only.\\
(2) For a random vector $(\xi,\eta)$, $H(\xi,\eta)\le H(\xi)+H(\eta)$; the equality holds iff $\xi$ and $\eta$ are independent.
\end{lemma}

\section{Entropy characteristics of prefix trees}\label{s:entropy}

Let $\cT=\cT(L)$ be a prefix tree. The entropy characteristics introduced below measure the expected uncertainty \emph{of a single letter} in a random word from $L$. By \emph{order-$n$ general entropy} $H_n(L)$ we mean the entropy of a random variable uniformly distributed on the set $|L\cap\Sigma^n|$ (or on the set of level-$n$ nodes of $\cT$), divided by $n$. By Lemma~\ref{l:shannon}(1), $H_n(L)=\frac{\log C_L(n)}{n}$. The fact that $L$ is factorial guarantees the existence of the limit
$$
H(L)=\lim_{n\to\infty} H_n(L)=\lim_{n\to\infty}\log (C_L(n))^{1/n}=\log \gr(L),
$$
which we call the \emph{general entropy} of $L$.

A different notion of entropy stems from consideration of random walks in $\cT$. As usual in graph theory, by \emph{random walk} we mean a stochastic process (Markov chain), the result of which is a finite or infinite walk in the given graph. The process starts in the initial state (either fixed or randomly chosen from some distribution) and runs step by step, guided by the following rule: visiting the node $u$, choose uniformly at random\footnote{Non-uniform distributions are also used in many applications but we do not consider them here.} an outgoing edge of $u$ and follow it to reach the next node. The walk stops if $u$ has no outgoing edges. Note that all walks in $\cT$ are directed paths; we refer to the walks starting at the root as \emph{standard}. 
Let $\eta_n$ be the random variable with the range $|L\cap\Sigma^n|$ such that the probability of a word $w\in L$ is the probability that a random standard walk in $\cT$, reaching the level $n$, visits $w$. The \emph{order-$n$ Markov entropy} of $L$ is $\mu_n(L)=\frac{H(\eta_n)}{n}$. The following lemma is immediate from definitions and Lemma~\ref{l:shannon}(1). 

\begin{lemma} \label{l:mar1} For any factorial language $L$
and any $n$, one has $\mu_n(L)\le H_n(L)$.
\end{lemma}

Similar to the case of the general entropy, the limit value exists:
\begin{lemma} \label{l:mar2} Let $L$ be a factorial language. Then 
there exists a limit $\mu(L)=\lim_{n\to\infty} \mu_n(L)=\inf_{n\in\mathbb{N}}\mu_n(L)$.
\end{lemma}
\begin{proof} Consider a random walk/word of length $n{+}m$ as a ``vector'' consisting of two random walks/words with lengths $n$ and $m$ respectively. Then $H(\eta_{n+m})\le H(\eta_n)+H(\eta_m)$ by Lemma~\ref{l:shannon}(2). Hence the sequence $\{H(\eta_n)\}$ is subadditive as a function of $n$, and Fekete's lemma \cite{Fek23} guarantees the existence of the limit $\lim_{n\to\infty} \frac{H(\eta_n)}{n}$ and its equality to the infimum, as required.\qed
\end{proof}

We call $\mu(L)$ the \emph{Markov entropy} of $L$. We want to estimate $\mu(L)$ for different languages; so our first goal is to relate $H(\eta_n)$, and thus $\mu_n(L)$, to the parameters of the tree $\cT$. Let $\ch(w)$ denote the number of children of the node $w$ in $\cT$ and $P(w)$ be the probability of visiting the word $w$ by a random standard walk. 

\begin{lemma} \label{l:wordprob} $P(w)=\Big(\prod_{i=0}^{|w|-1} \ch(w[1..i]) \Big)^{-1}$.
\end{lemma}
\begin{proof} 
All prefixes of $w$ should be visited and the right choice should be done at each step.\qed
\end{proof}

In general, $P(w)$ may underestimate the probability assigned to $w$ by $\eta_{|w|}$; this is the case if some prefix of $w$ has a child which generates a finite subtree with no nodes of level $|w|$. To remedy this, we consider trimming of prefix trees. By \emph{$n$-trimmed version} of $\cT$, denoted by $\cT_{[n]}$, we mean the tree obtained from $\cT$ by deletion of all finite subtrees $\cT_u$ which have no nodes of level $n$ (and thus of bigger levels). In other words, a node $w\in L$ is deleted iff $L$ contains no length-$n$ word with the prefix $w$. 

\begin{example}
Let $L=\PF(2,3)$, $\cT=\cT(L)$. If $n\ge9$, then $\cT_{[n]}$ does not contain $u=00100100$ because $u0,u1$ end with cubes; if $n\ge 15$, then  $\cT_{[n]}$ does not contain $v=0100101001010$, because $v1,v00$, and $v01$ end with cubes.
\end{example}

The \emph{trimmed version} of $\cT$, denoted by $\cT_{[]}$, is obtained from $\cT$ by deletion of \emph{all} finite subtrees. The next lemma follows from definitions.
\begin{lemma} \label{l:trim}
(1) $\cT_{[]}=\bigcap_{n\in\mathbb{N}} \cT_{[n]}$. (2) $\cT_{[]}$ is the prefix tree of $\re(L)$.
\end{lemma}
We write $\ch_{[n]}(w)$ ($\ch_{[]}(w)$) for the number of children of $w$ in $\cT_{[n]}$ (resp., $\cT_{[]}$) and $P_{[n]}(w)$ ($P_{[]}(w)$) for the probability of visiting $w$ by a random standard walk in $\cT_{[n]}$ (resp., $\cT_{[]}$). As in Lemma~\ref{l:wordprob}, one has
\begin{equation} \label{e:Pn}
    P_{[n]}(w)=\Big(\prod_{i=0}^{|w|-1} \ch_{[n]}(w[1..i]) \Big)^{-1} \text{ and \ } P_{[]}(w)=\Big(\prod_{i=0}^{|w|-1} \ch_{[]}(w[1..i]) \Big)^{-1} .
\end{equation}

\begin{lemma} \label{l:probtrim}
Let $w\in L$, $|w|=n$. Then $\eta_n$ assigns to $w$ the probability $P_{[n]}(w)$.
\end{lemma}
\begin{proof}
In $\cT_{[n]}$, one can perform $n$ steps of a random walk and reach the level $n$ no matter which random choices were made. Hence  the probability assigned to $w$ by $\eta_n$ equals $P_{[n]}(w)$ by definition.\qed
\end{proof}

Lemma~\ref{l:probtrim} and the definition of entropy imply
\begin{equation} \label{e:entr}
    H(\eta_n)=-\sum_{w\in L\cap\Sigma^n} P_{[n]}(w)\log P_{[n]}(w)
\end{equation}

Given an arbitrary tree $\cT$, we assign to each internal node $u$ its \emph{weight}, equal to the logarithm of the number of children of $u$. \emph{Branching frequency} of standard walk ending at a node $w$, denoted by $\Bf(\cT,w)$, is the sum of weights of all nodes in the walk, except for $w$, divided by the length of the walk (= level of $w$). The use of branching frequency for prefix trees can be demonstrated as follows. For a language $L$, a natural problem is to design a method for compact representation of an arbitrary word $w\in L$. A possible solution is to encode the standard walk in $\cT=\cT(L)$, ending at $w$. We take $|w|$-trimmed version of $\cT$ and encode consecutively all choices of edges needed to reach $w$. For each predecessor $u$ of $w$ we encode the correct choice among $\ch_{[|w|]}(u)$ outgoing edges. The existence of asymptotically optimal entropy coders, like the arithmetic coder \cite{Ris76}, allows us to count $\log \ch_{[|w|]}(u)$ bits for encoding this choice. Thus $w$ will be encoded by $\sum_{i=0}^{|w|-1} \log \big(\ch_{[|w|]}(w[1..i])\big)$ bits, which is exactly $\Bf(\cT_{[|w|]},w)$ bits per symbol.  

\begin{remark} \label{r:binary}
The proposed method of coding generalizes the antidictionary compression method \cite{CMRS99} for arbitrary alphabets. Antidictionary compression works as follows: given $w\in L\subseteq\{0,1\}^*$, examine each prefix $w[1..i]$; if it is the only child of $w[1..i{-}1]$ in the prefix tree of $L$, delete $w[i]$. In this way, the remaining bits encode the choices made in branching points during the standard walk to $w$. The compression ratio is the fraction of branching points among the predecessors of $w$: any branching point contributes 1 to the length of the code, other nodes in the walk contribute nothing.
\end{remark}

The following theorem relates branching frequencies to Markov entropy.
\begin{theorem} \label{t:bf}
For a factorial language $L$ and a positive integer $n$, the order-$n$ Markov entropy of $L$ equals the expected branching frequency of a length-$n$ random walk in the prefix tree $\cT(L)$.
\end{theorem}
\begin{proof}
Let $E_n$ denote the expected branching frequency of a length-$n$ random walk in $\cT(L)$. 
One has
\begin{multline*}
   E_n=[\text{definition of expectation}] =\!\!\sum_{w\in L\cap\Sigma^n}\!\! P_{[n]}(w)\Bf(\cT_{[n]},w) \\ =[\text{definition of  }\Bf] 
   = \!\!\sum_{w\in L\cap\Sigma^n}\!\! P_{[n]}(w)\Big(\sum_{i=0}^{|w|-1} \log \big(\ch_{[n]}(w[1..i])\big)\Big)/n \\
   = [\text{Eq. \eqref{e:Pn}}] = \!\!\sum_{w\in L\cap\Sigma^n}\!\! P_{[n]}(w)\big(- \log P_{[n]}(w)\big)/n 
   = [\text{Eq. \eqref{e:entr}}] = \frac{H(\eta_n)}{n} = \mu_n(L).\ \qed
\end{multline*}

\end{proof}

\section{Computing Entropy} \label{s:computeentropy}
\subsection{General and Markov entropy for regular languages} \label{ss:reg}

Let $L$ be a factorial regular language, $\cA$ be a PDFA, recognizing $L$. The problem of finding $\gr(L)$, and thus $H(L)$, was solved by means of matrix theory. Let us recall the main steps of this solution. The Perron--Frobenius theorem says that the maximum absolute value of an eigenvalue of a non-negative matrix $M$ is itself an eigenvalue, called \emph{principal} eigenvalue. A folklore theorem (see Theorem~2 in the survey \cite{Sh12rev}) says that $\gr(L)$ equals the principal eigenvalue of the adjacency matrix of $\cA$. This eigenvalue can be approximated\footnote{Note that it is not possible in general to find the roots of polynomials exactly.} with any absolute error $\delta$ in $O(|\cA|/\delta)$ time \cite[Th.\,5]{Sh08csr}; see also \cite[Sect.\,3.2.1]{Sh12rev}.

Now consider the computation of $\mu(L)$. By Lemma~\ref{l:mar2} and Theorem~\ref{t:bf}, $\mu(L)$ is the limit of expected branching frequencies of length-$n$ random standard walks in the prefix tree $\cT=\cT(L)$. Standard walks in $\cT$ are in one-to-one correspondence with accepting walks in $\cA$, so we can associate each node $w\in \cT$ with the state $\lambda.w\in \cA$ and consider random walks in $\cT$ as random walks in $\cA$. We write $\deg^\to(u)$ for the out-degree of the node $u$ in $\cA$.

We need the apparatus of \emph{finite-state Markov chains}. Such a Markov chain with $m$ states is defined by an arbitrary row-stochastic $m\times m$ matrix $A$ (row-stochastic means that all entries are nonnegative and all row sums equal 1). The value $A[i,j]$ is treated as the probability that the next state of the chain will be $j$ given that the current state is $i$. Any finite directed graph $G$ with no nodes of out-degree 0 represents a finite-state Markov chain; the stochastic matrix of $G$ is built as follows: take the adjacency matrix and divide each value by the row sum of its row (see Fig.~\ref{f:2nodes} below). 

Recall some results on finite-state Markov chains (see, e.g., \cite[Vol.\,2, Ch.\,3]{Gan60}). Let $A$ be the $m\times m$ matrix of the chain. The process is characterized by the vectors $\vec p^{(n)}=(p_1^{(n)},\ldots, p_m^{(n)})$, where $p_i^{(n)}$ is the probability of being in state $i$ after $n$ steps; the initial distribution $\vec p^{(0)}$ is given as a part of description of the chain. \emph{Stationary distribution} of $A$ is a vector $\vec p=(p_1,\ldots,p_m)$ such that $p_i\ge0$ for all $i$, $\sum_{i=1}^m p_i=1$ and $\vec pA=\vec p$. Every row-stochastic matrix has one or more stationary distributions; such a distribution is unique for the matrices obtained from strongly connected digraphs. The sequence $\{\vec p^{(n)}\}$ approaches some stationary distribution $\vec p$ in the following sense:
\begin{itemize}
    \item[$(*)$] there exists an integer $h\ge1$ such that $\vec p=\lim_{n\to\infty} \frac{\vec p^{(n)}+\vec p^{(n+1)}+\ldots+\vec p^{(n+h-1)}}{h}$
\end{itemize} 
(That is, the limit of the process is a length-$h$ cycle and $\vec p$ gives the average probabilities of states in this cycle. In practical cases usually $h=1$ and thus $\vec p = \lim_{n\to\infty} \vec p^{(n)}$.)

\begin{theorem} \label{t:reg}
Let $L$ be a factorial regular language, $\hat{\cA}$ be a PDFA accepting $\re(L)$. Suppose that $\hat{\cA}$ has $m$ states $1,\ldots,m$ and $\vec p=(p_1,\ldots,p_m)$ is the stationary distribution for a random walk in $\hat{\cA}$, starting at the initial state. Then $\mu(L)=\sum_{i=1}^m p_i\log(\deg^\to(i))$.
\end{theorem}
\begin{proof}
In $\hat{\cA}$, all nodes have outgoing edges, so $\hat{\cA}$ defines a finite-state Markov chain. As above, $p_i^{(n)}$ denotes the probability to be in the state $i$ after $n$ steps of a random walk. Let 1 be the initial state; then $p_1^{(0)}=1$ and $p_i^{(0)}=0$ for $i=2,\ldots,m$ by conditions of the theorem. By linearity of expectation, the expected number of visits to state $i$ before the $n$th step is $N(i,n)=\sum_{j=0}^{n-1} p_i^{(j)}$. By $(*)$, $\lim_{n\to\infty} \frac{p_i^{(0)}+\cdots+p_i^{(n-1)}}{n}=p_i$ and thus $N(i,n)\xrightarrow{n\to\infty} np_i$. 

Let $\hat\mu(L)$ be the Markov entropy of $\re(L)$. Consider a word $w\in\re(L)$ and the corresponding standard walk in the prefix tree $\hat\cT=\cT(\re(L))$. In $\hat\cT$, all walks can be infinitely extended: $\ch(u)=\ch_{[n]}(u)$ for all $u,n$. Hence $\Bf(\hat\cT,w)=\frac 1n\sum_{i=0}^{|w|-1} \log \big(\ch(w[1..i])\big)$. Further, $\ch(u)$ equals the out-degree of the state $\lambda.u$ of $\hat{\cA}$. Hence, the expected branching frequency $\hat E_n$ of a length-$n$ random standard walk in $\hat\cT$ satisfies
\begin{equation} \label{e:En}
\hat E_n= \frac 1n \sum_{i=1}^m N(i,n)\log(\deg^\to(i))    
\end{equation}
Taking limits as $n$ approaches infinity and substituting $np_i$ for $N(i,n)$, one gets
\begin{equation} \label{e:hatmu}
\hat\mu(L) = \frac 1n \sum_{i=1}^m np_i\log(\deg^\to(i)) = \sum_{i=1}^m p_i\log(\deg^\to(i)).
\end{equation}
It remains to show that $\mu(L)=\hat\mu(L)$. Let $L$ be accepted by a PDFA $\cA$. For $w\in L$, one trivially has
\begin{itemize}
    \item[$(\diamond)$] $w\in\re(L)$ iff some cycle of $\cA$ can be reached from the state $\lambda.w$
\end{itemize}
By $(\diamond)$, the PDFA obtained from $\cA$ by deleting all states from which no cycle is reachable accepts $\re(L)$. So we may refer to the PDFA obtained from $\cA$ by such ``trimming'' as $\hat\cA$. The number of deleted states is some constant $c$. Consider an accepting walk in $\cA$; let $w$ be its label. The first $|w|-c$ steps of this walk are made inside $\hat \cA$. Then we use \eqref{e:En} to estimate the expected branching frequency $E_n$ of a length-$n$ random standard walk in the tree $\cT=\cT(L)$:
$$
\frac 1n \sum_{i=1}^m N(i,n{-}c)\log(\deg^\to(i))\le E_n\le \frac 1n \big(c\log\sigma+ \sum_{i=1}^m N(i,n{-}c)\log(\deg^\to(i))\big).
$$
The lower (resp., upper) bound is obtained by replacing the sum of weights of the last $c$ steps of the walk with the lower bound 0 (resp., upper bound $c\log\sigma)$. As $n$ tends to infinity, both the lower and the upper bound approach the expression from the right-hand side of \eqref{e:hatmu}. Hence $\mu(L)=\hat\mu(L)$; the theorem is proved.\qed
\end{proof}

\begin{example} \label{ex:2nodes}
Let $L\subset\{0,1\}^*$ be the regular language consisting of all words having no factor 11. Its accepting PDFA, the corresponding matrices and computations are presented in Fig.~\ref{f:2nodes}. Note that $\re(L)=L$.
\begin{figure}[!htb]
    \begin{minipage}{3.8cm}
    \includegraphics[trim = 43 762 450 31, clip]{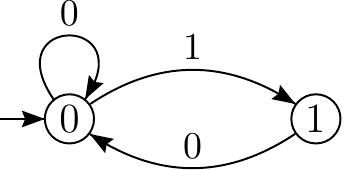}
    \end{minipage}
    \begin{minipage}{4.2cm}
    Adjacency matrix: $\begin{bmatrix} 1&1\\ 1&0\end{bmatrix}$\\
    Characteristic polynomial:\\
    $x^2-x-1$\\
    $\gr(L)=\frac{\sqrt{5}+1}{2}\approx 1.618\ldots$\\
    $H(L)=\log (\gr(L)){\approx} 0.694\ldots$
    \end{minipage}
    \begin{minipage}{4.0cm}
    Stochastic matrix: $\begin{bmatrix} \frac 12&\frac 12\\ 1&0\end{bmatrix}$\\
    Row eigenvector: $\vec p=(\frac 23, \frac 13)$\\
    $\mu(L)=\frac 23\log 2 = \frac 23$
    \end{minipage}
    
    \caption{Accepting PDFA and entropy computations for the language $L$ (Example~\ref{ex:2nodes}).}
    \label{f:2nodes}
\end{figure}
\vspace*{-3mm}
\end{example}
\paragraph{Computational aspects.} Computing $\hat\cA$ from $\cA$ takes $O(|\cA|)$ time, as it is sufficient to split $\cA$ into strongly connected components and traverse the acyclic graph of components. The vector $\vec p$ can be computed by solving the linear system $\vec p(\hat A-I)=\vec 0$, where $\hat A$ is the adjacency matrix of $\hat\cA$ and $I$ is the identity matrix. This solution requires $\Theta(m^3)$ time and $\Theta(m^2)$ space, which is forbidding for big automata. More problems arise if the solution is not unique; but the correct vector $\vec p$ still can be found by means of matrix theory (see, e.g., \cite[Ch.\,11]{GrSn97}). In order to process big automata (say, with millions of states), one  can iteratively use the equality $\vec p^{(n+1)}=\vec p^{(n)}\!\hat A$ to approximate $\vec p$ with the desired precision. Each iteration can be performed in $O(m)$ time, because $\hat A$ has $O(m)$ nonzero entries. One can prove, similar to \cite[Th.\,3.1]{Sh10tc}, that under certain natural restrictions $O(\delta^{-1})$ iterations is sufficient to obtain $\vec p$ within the approximation error $\delta$.

\subsection{Order-$n$ Markov entropy via random walks} \label{ss:walks}

Let $L\subseteq\Sigma^*$ be an arbitrary infinite factorial language such that the predicate $\cL(w)$, which is true if $w\in L$ and false otherwise, is computable. There is little hope to compute $\mu(L)$, but one can use an oracle computing $\cL(w)$ to build random walks in the prefix tree $\cT=\cT(L)$ and obtain statistical estimates of $\mu_n(L)$ for big $n$. We construct random walks by random depth-first search (Algorithm~\ref{alg:dfs}), executing the call $\dfs(\lambda,n)$. The algorithm stops immediately when level $n$ is reached. When visiting node $u$, the algorithm chooses a non-visited child of $u$ uniformly at random and visits it next. If all children of $u$ are already visited, then $u$ is a ``dead end'' (has no descendants at level $n$), and the search returns to the parent of $u$.
\begin{algorithm*}[!htb]
\caption{Random walk in $\cT(L)$ by depth-first search} 
\label{alg:dfs}
\begin{algorithmic}[1]
\State{\textbf{function} $\dfs(u,n)$ \Comment{$u$=node, $n$=length of walk}}
    \If {$|u|=n$} break \Comment{walk reached level $n$}
    \EndIf
\State {$(a_1a_2\ldots a_\sigma)\gets$ random permutation of $\Sigma$}
\For {$j=1$ to $\sigma$}
    \If {$\cL(ua_j)$} $\dfs(ua_j,n)$\Comment{visit $ua_i$ next}
    \EndIf
\EndFor
\State {return \Comment{$u$ has no descendant at level $n$}}
\end{algorithmic}
\end{algorithm*}
\vspace*{-7mm}
\begin{lemma} \label{l:dfs}
$\dfs(\lambda,n)$ builds a length-$n$ random standard walk in $\cT(L)$.
\end{lemma}
\begin{proof}
Let $w$ be the last node of the walk built, $u$ be an arbitrary prefix of $w$. Among the children of $u$ in $\cT(L)_{[n]}$, all had equal probabilities to be chosen earlier than the others. The choice of a child of $u$ in $\cT(L)$ which does not belong to $\cT(L)_{[n]}$ would affect neither the result of search nor the probabilities of other children to be chosen next. Thus, the obtained walk is random by definition.\qed
\end{proof}

Consider the values of the counter $j$ in the instances $\dfs(\lambda,n),\dfs(w[1],n)$, $\ldots,\dfs(w[1..n{-}1],n)$ at the moment when the search reaches level $n$. We denote \emph{profile} of the constructed walk as the vector $\vec r=(r_1,\ldots,r_\sigma)$ such that $r_i$ is the number of instances of $\dfs$ in which $j=i$. Note that different runs of Algorithm~\ref{alg:dfs} may result in the same walk with different profiles (due to random choices made, depth-first search visits some dead ends and skips some others). Given a profile $\vec r$, one can compute the expected branching frequency $\Bf(\vec r)$ of a walk with this profile: $\Bf(\vec r)=\frac 1n\sum_{i=1}^\sigma c_i\log i$, where the parameters $c_i$ are computed in Theorem~\ref{t:profile} below.

\begin{theorem} \label{t:profile}
Let $\vec r=(r_1,\ldots,r_\sigma)$ be a profile of a length-$n$ random standard walk in a tree $\cT$.
For each $i=1,\ldots,\sigma$, let $c_i$ be the expected number of nodes, having exactly $i$ children in the tree $\cT_{[n]}$, in a random standard walk with the profile $\vec r$. Then 
\begin{equation} \label{e:profile}
(c_1,\ldots,c_\sigma)P=\vec r, 
\text{ where }
P[i,k]=\frac{\binom{\sigma-i}{k-1}}{\binom{\sigma}{k-1}}-
\frac{\binom{\sigma-i}{k}}{\binom{\sigma}{k}}
\text{ for }i,k=1,\ldots,\sigma.    
\end{equation}
\end{theorem}
\begin{proof}
Each call to $\dfs(u,n)$, made during a length-$n$ random standard walk with the profile $\vec r$ in $\cT$, contributes 1 to one of the numbers $r_1,\ldots,r_{\sigma}$, depending on the random permutation generated within this call. Let $p_{ik}$ be the probability that a call for a node with $i$ children in the tree $\cT_{[n]}$ contributes to $r_k$. Then the expected contribution of such a call to $\vec r$ is $(p_{i1},\ldots,p_{i\sigma})$. The sum of expected contributions of all calls ($c_1$ calls for nodes with 1 child, \ldots, $c_\sigma$ calls for nodes with $\sigma$ children) should be equal to $\vec r$: $r_k=\sum_{i=1}^\sigma c_i p_{ik}$ for $k=1,\ldots,\sigma$. Hence $\vec r=(c_1,\ldots,c_\sigma)P$, where $P[i,k]=p_{ik}$. To compute $p_{ik}$, note that contributing to $r_k$ means that $\dfs(u,n)$ executed $k-1$ unsuccessful calls to $\dfs$ before a successful call. Since there are $\sigma-i$ choices for an unsuccessful call, $p_{ik}$ equals the difference between $\binom{\sigma-i}{k-1}/\binom{\sigma}{k-1}$ (the probability that the first $(k{-}1)$ calls are unsuccessful) and $\binom{\sigma-i}{k}/\binom{\sigma}{k}$ (the probability that the first $k$ calls are unsuccessful).\qed
\end{proof}
\begin{example} \label{ex:profiles}
Let us solve \eqref{e:profile} for $\sigma=2$ (left) and $\sigma=3$ (right):\\
\begin{minipage}{4.5cm} 
$(c_1,c_2) \begin{bmatrix}
\frac 12&\frac 12\\ 1&0\end{bmatrix} =(r_1,r_2)$\\[2mm]
$c_1=2r_2$, $c_2=r_1-r_2$\\
$\Bf(\vec r)=\frac{r_1-r_2}{r_1+r_2}$
\end{minipage}
\begin{minipage}{7cm} 
$(c_1,c_2,c_3)\begin{bmatrix}
\frac 13&\frac 13&\frac 13\\ \frac 23&\frac 13&0\\ 1&0&0\end{bmatrix} =(r_1,r_2,r_3)$\\
$c_1=3r_3$, $c_2=3r_2-3r_3$, $c_3=r_1-2r_2+r_3$\\
$\Bf(\vec r)=\frac{3(r_2-r_3)+(r_1-2r_2+r_3)\log 3}{r_1+r_2}$
\end{minipage}
\end{example}

\section{Experimental Results} \label{s:exper}
\subsection{Regular approximations of power-free languages} \label{ss:regappr}

This was a side experiment in the comparison of general entropy and Markov entropy for power-free languages. We took the ternary square-free language $\SF=\PF(3,2)$, which is a well-studied test case. Its growth rate $\gr(\SF)\approx 1.30176$ is known with high precision \cite{Sh12rev} from the study of its \emph{regular approximations}. A $k$th regular approximation $\SF_k$ of $\SF$ is the language of all words having no squares of period $\le k$ as factors. The sequence $\{\gr(\SF_k)\}$ demonstrates a fast convergence to $\gr(\SF)$. So we wanted to (approximately) guess the Markov entropy $\mu(\SF)$ extrapolating the  initial segment of the sequence $\mu(\SF)_k$.

The results are as follows: we computed the values $\mu(\SF_k)$ up to $k=45$ with absolute error $\delta<10^{-8}$ using the technique from Section~\ref{ss:reg}. We obtained $\mu(\SF_{45})\approx0.36981239$; the extrapolation of all obtained values gives $0.369810<\mu(\SF)<0.369811$. At the same time we have $H(\SF)=\log(\gr(\SF))\approx 0.380465$, so the two values are clearly distinct but close enough.

\subsection{Random walks in power-free languages} \label{ss:expPF}

To perform experiments with length-$n$ random walks for a language $L$, one needs an algorithm to compute $\cL(w)$ to be used with Algorithm~\ref{alg:dfs}. A standard approach is to maintain a data structure over the current word/walk $w$, which quickly answers the query ``$w\in L$?'' and supports addition/deletion of a letter to/from the right. The theoretically best such algorithm for power-free words was designed by Kosolobov \cite{Kos15}: it spends $O(\log n)$ time per addition/deletion and uses memory of size $O(n)$. However, the algorithm is complicated and the constants under $O$ are big. We developed a practical algorithm which is competitive for the walks up to the length of several millions. For simplicity, we describe it for square-free words but the construction is the same for any power-free language. 

We use arrays $\rep_i[1..n]$, $i=0,\ldots,\lfloor \log n\rfloor-1$ to store previous occurrences of factors. Namely, if $|u|\ge j$ for the current word $u$, then $\rep_i[j]$ is the \emph{last} position of the previous occurrence of the factor $u[j{-}2^i{+}1..j]$ or $-\infty$ if there is no previous occurrence. Deletion of a letter from $u$ is performed just by deleting the entries $\rep_i[|u|]$; let us consider the procedure $\add(u,a)$ (Algorithm~\ref{alg:2free}) which adds the letter $a$ to $u$, checks square-freeness of $ua$ and computes $\rep_i[|u|+1]$. The auxiliary array $\last[1..\sigma]$ stores the rightmost position of each letter in the current word.

\begin{algorithm*}[!htb]
\caption{Online square detection: adding a letter} 
\label{alg:2free}
\begin{algorithmic}[1]
\State{\textbf{function} $\add(u,a)$ \Comment{$u$=word, $a$=letter to add}}
\State{$\rep_0[|u|+1]\gets \last[a]$; $\last[a]\gets |u|+1$ \Comment{fill previous occurrence of $a$}}
\State{$\free\gets \mathsf{true}$ \Comment{square-freeness flag}}
\For {$i=0$ to $\lfloor \log n\rfloor-1$}
    \State{$x\gets \rep_i[|u|{+}1]$; $p=|u|+1-x$ \Comment{$p$ is the possible period of a square}}
    \If {$p\le 2^{i+1}$ and $\rep_i[x{+}2^i] = x+2^i -p$} $\free\gets \mathsf{false}$; break \Comment{Fig.~\ref{f:detect}}
    \EndIf
    \If {$i=\lfloor \log n\rfloor-1$} break \Comment{no more arrays to update}
    \EndIf
    \State{compute $\rep_{i+1}[|u|+1]$ \Comment{from $\rep_i$}}
    \If {$\rep_{i+1}[|u|+1]=-\infty$} break \Comment{all repeated suffixes processed}
    \EndIf
\EndFor
\State {return $\free$ \Comment{the answer to ``is $ua$ square-free?''}}
\end{algorithmic}
\end{algorithm*}

\begin{wrapfigure}[6]{r}{6cm}
\centering
    \includegraphics[trim = 45 768 380 37, clip]{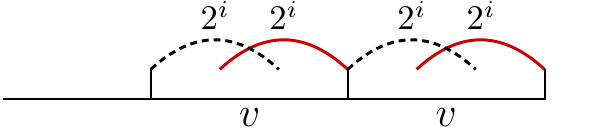}
    \caption{\vspace*{-3mm} Detecting a square by Algorithm~\ref{alg:2free}.}
\label{f:detect}
\end{wrapfigure}
\paragraph*{Correctness.} Recall that $u$ is square-free, so the occurrences of a factor of $u$ can neither overlap nor touch. Assume that $ua$ ends with a square $vv$, $p=|v|$, $2^i<p\le 2^{i+1}$. Then $p$ will be found in line~5 as $|u|+1-\rep_i[|u|{+}1]$ (red arcs in Fig.~\ref{f:detect} show the suffix of length $2^i$ and its previous occurrence). The condition in line~6 means exactly the equality of words marked by dash arcs in Fig.~\ref{f:detect}; thus, $vv$ is detected and $\add(u,a)$ returns $\mathsf{false}$. For the other direction, if $\add(u,a)=\mathsf{false}$, then the condition in line~6 held true and thus a square was detected as in Fig.~\ref{f:detect}.

\paragraph*{Time complexity.} Everything except line~8 takes $O(\log n)$ time. Let $vz$ be the suffix of $ua$ such that $|v|=|z|=2^i$. To find the previous occurrence of $vz$, we scan in $\rep_i$ the occurrences of $z$ right to left, looking for an occurrence preceded by $v$ (the occurrences of $v$ are also scanned right to left). On expectation, we will check $O(s_{|z|})$ occurrences of $z$ before finding $vz$, where $s_k$ is the number of $\sigma$-ary square-free words of length $k$. On the other hand, the expected number of occurrences of $z$ in $ua$ is $O(n/s_{|z|})$. The same bounds apply to $v$. Hence the total number of array entries accessed in line~8 during one call to $\add$ is, on expectation, $\sum_{i=0}^{\lfloor \log n\rfloor-1} \min\big\{O(s_{2^i}),O(\frac n{s_{2^i}})\big\}$. As $s_k$ grows exponentially with $k$, this sum is $O(\sqrt{n})$. Hence the expected time for one call to $\add$ is $O(\sqrt{n})$. The experiments confirm this estimate.

\paragraph*{Experiments.} We studied the following languages: binary cube-free and ternary square-free languages as typical ``test cases'', threshold languages over 3,\ldots,10, 20, 50, and 100 letters, and $\PF(2,{\frac73}^+)$ as the smallest binary language of exponential growth. All non-binary languages from this list have ``essentially binary'' prefix trees: a letter cannot coincide with any of $(\sigma{-}2)$ preceding letters, and so any node of level at least $\sigma{-}2$ has at most two children. Hence we computed expected branching frequencies of walks as in Example~\ref{ex:profiles}. For each of the languages we computed profiles of 1000 walks of length $10^5$ and 100 walks of length $10^6$. For the tables with the obtained data see \cite{Stat}. We briefly analysed the data. The most interesting findings, summarized below, are the same for each of the studied languages. Some figures are presented in Table~\ref{tab:pf}.\\
\textbf{1}. The profiles of all walks in $\cT=\cT(L)$ are close to each other. To be precise, assume that $\Bf(\vec r)=\mu_n(L)$ for all constructed profiles. Then the number $r_2$ computed for a length-$n$ random walk is the number of heads in $c_1$ tosses of a fair coin (among $c_1$ nodes with two children, in $r_2$ cases the dead end was chosen first). Hence the computed values of $r_2$ form a sample from the binomial distribution $B(c_1,\frac{1}{2})$. And indeed, the set of computed $r_2$'s looks indistinguishable from such a sample; see \cite[stat100000]{Stat}. This property suggests the mean value of $\Bf(\vec r)$ over all experiments as a good approximation of $\mu_n(L)$.\\
\textbf{2}. The 99\% confidence interval for the mean branching frequency $\Bf(\cT_[|w|],w)$ of the 1000 constructed walks of length $10^5$ is of length $\sim 4\cdot 10^{-4}$ and includes the mean value of $\Bf(\vec r)$ for the walks of length $10^6$. For the language $\SF$, this interval also includes the value $\mu(\SF)$ conjectured in Section~\ref{ss:regappr}. This property suggests that $\mu_n(L)$ for such big $n$ is close to the Markov entropy $\mu(L)$.\\
\textbf{3}. As $\mu(L)\le H(L)=\log(\gr(L))$, the value of $\mu(L)$ can be converted to the lower bound for the growth rate of $L$. The values $2^{\textrm{mean}(\Bf(\vec r))}$ from our experiments differ from the best known upper bounds for the studied languages (see \cite[Tbl.\,A1-A3]{Sh12rev}) by the margin of 0.004-0.018. Such a bound is quite good  for all cases where specialized methods \cite{KoRa11,Sh09dlt} do not work. The results for threshold languages support the Shur--Gorbunova conjecture \cite{ShGo10} that the growth rates of these languages tend to the limit $\alpha\approx 1.242$ as the size of the alphabet approaches infinity.
\vspace*{-5mm}
\begin{table}[!htb]
$$    \begin{array}{|c|c|c|c|c|}
    \hline
\text{Language}&\text{Mean } \Bf(\vec r)\ (10^5)&
\text{Mean } \Bf(\vec r)\ (10^6)&2^{\mu(L)}& \gr(L)\\
\hline
\PF(2,{\frac 73}^+)&0.27221&0.27220&\approx1.20766&\approx 1.22064\\
\PF(2,3)&0.52562&0.52553&\approx1.43956&\approx 1.45758\\
\PF(3,{\frac 74}^+)&0.30249&0.30251&\approx1.23327&\approx 1.24561\\
\PF(3,2)&0.36988&0.36987&\approx1.29223&\approx 1.30176\\
\PF(4,{\frac 75}^+)&0.09137&0.09151&\approx1.06535&< 1.06951\\
\PF(5,{\frac 54}^+)&0.20279&0.20265&\approx1.15092&< 1.15790\\
\PF(6,{\frac 65}^+)&0.28536&0.28526&\approx1.21871&< 1.22470\\
\PF(7,{\frac 76}^+)&0.29753&0.29749&\approx1.22903&< 1.23690\\
\PF(8,{\frac 87}^+)&0.28881&0.28867&\approx1.22163&< 1.23484\\
\PF(9,{\frac 98}^+)&0.30716&0.30732&\approx1.23727&< 1.24668\\
\PF(10,{\frac {10}9}^+)&0.29674&0.29669&\approx1.22836&< 1.23931\\
\PF(20,{\frac {20}{19}}^+)&0.30002&0.29982&\approx1.23099&< 1.24205\\
\PF(50,{\frac {50}{49}}^+)&0.30006&0.29970&\approx1.23089&< 1.24210\\
\PF(100,{\frac {100}{99}}^+)&0.30047&0.29974&\approx1.23093&< 1.24210\\
\hline
    \end{array} $$
    \caption{Markov entropy for power-free languages: experiments}
    \label{tab:pf}
\end{table}
\vspace*{-9mm}
\subsection{Random walks in Abelian power-free languages}

Similar to Section~\ref{ss:expPF}, we need an algorithm checking Abelian power-freeness. Here we describe an algorithm detecting Abelian squares; its modification for other integer powers is straightforward. If a word $w[1..n]$ is fixed, we  let $\vec V_i=$ $\vec V(w[n{-}i{+}1..n])-\vec V(w[n{-}2i{+}1..n{-}i])$. A simple way to find whether $w$ ends with Abelian square is to check 
$\vec V_i=\vec0$ for all $i$. Since $\vec V_{i+1}$ can be obtained from $\vec V_i$ with a constant number of operations (add $w[n{-}i]$ twice, subtract $w[n{-}2i]$ and $w[n{-}2i{-}1]$), this check requires $\Theta(n)$ time. However, $\Theta(n)$ time per iteration appeared to be too much to perform experiments comparable with those for power-free languages, so we developed a faster algorithm. It maintains two length-$n$ arrays for each letter $a\in\Sigma$: $d_a[i]$ is the position of $i$th from the left letter $a$ in the current word $w$ and $c_a[i]$ is the number of occurrences of $a$ in $w[1..i]$ (i.e., a coordinate of $\vec V(w[1..i])$). When a letter is added/deleted, these arrays are updated in $O(1)$ time (we regard $\sigma$ as a constant). The function $\asq(u)$ (Algorithm~\ref{alg:a2free}) checks whether the word $w$ has an Abelian square as a suffix. 
\begin{algorithm*}[!htb]
\caption{Online Abelian square detection} 
\label{alg:a2free}
\begin{algorithmic}[1]
\State{\textbf{function} $\asq(u)$ \Comment{$u$=word}}
\State{$l\gets |u|-1$; $i\gets 1$ \Comment{two counters}}
\State{$\free\gets \mathsf{false}$ \Comment{square-freeness flag; turns true when check finishes}}
\While {\textbf{not} $\free$}
    \For {$a\in\Sigma$}
        \State{$\vec V[a]\gets c_a[|u|]-c_a[|u|{-}i]$ \Comment{$a$-coordinate of $\vec V(u[|u|{-}i{+}1..|u|])$}}
        \If {$\vec V[a]> c_a[|u|{-}i]$} $\free\gets \mathsf{true}$; break \Comment{no squares possible}
        \EndIf
        \State{$l \gets \min\{l, d_a[c_a[|u|{-}i]-\vec V[a]+1]\}$}
    \EndFor
    \If {$l=|u|-2i+1$} break \Comment{$u[|u|-2i+1..|u|]$ is an Abelian square}
    \EndIf
    \State{$i\gets \lceil(|u|-l+1)/2\rceil$ }
\EndWhile
\State {return $\free$ \Comment{the answer to ``is $ua$ Abelian square-free?''}}
\end{algorithmic}
\end{algorithm*}
\paragraph*{Correctness.} We show by induction the following property of Algorithm~\ref{alg:a2free}: at the beginning of each \textbf{while} cycle iteration, the word $u[l{+}1..|u|]$ is known to be Abelian square-free. The base case is obvious: $l{+}1=|u|$ before the first iteration. For the step case, assume that the property holds for some iteration. The value of $i$ assigned in line~10 satisfies $|u|{-}2i{+}1\le l < |u|{-}2i{+}3$. Hence the suffix $vz$ of $u$ with $|v|=|z|=i$ is the shortest one that can be an Abelian square. During the \textbf{for} cycle of the next iteration, $l$ receives the maximum value such that the word $Zz=u[l..|u|]$ satisfies $\vec V(z)\le \vec V(Z)$  (coordinate-wise). If $|Z|=i$, an Abelian square is found (line 9); otherwise, all proper suffixes of $Zz$ are Abelian square-free, which is exactly the property we are proving. Finally, if no such $Z$ exists (it was detected in line~7 that some letter $a$ occurs in $z$ more often than in the remaining part of $u$) then $u$ is proved Abelian square-free. Indeed, all longer suffixes of $u$ contain $z$ and thus more than half occurrences of $a$.
\paragraph*{Time Complexity.} In our experiments, Algorithm~\ref{alg:a2free} checked, on average, about $2\sqrt{n}$ suffixes of a length-$n$ word, but we have no theoretical proof of this fact.
\paragraph*{Experiments.} The structure and growth of Abelian-power-free languages are little studied. We considered the 4-ary Abelian-square-free language $\ASF$, the ternary Abelian-cube-free language $\ACF$, and the binary Abelian-4-power-free language $\A4F$; see Table~\ref{tab:apf}. Our main interest was in estimating the actual growth rate of these languages. The upper (resp. lower) bounds for the growth rates are taken from \cite{SaSh12} (resp., from \cite{Car98,Cur04,ACR04}). 
For $\ASF$ and $\ACF$ we got profiles of 500 walks of length $10^5$ and 100 walks of length $5{\cdot}10^5$; for $\A4F$, 100 profiles of walks of length $10^5$ were computed. The results suggest that the automata-based upper bounds for the growth rates of Abelian-power-free languages are quite imprecise, in contrast with the case of power-free languages. In addition, the experiments discovered the existence of very big finite subtrees on relatively low levels, which slow down the depth-first search. In fact, to obtain long enough words from $\A4F$ we modified the $\dfs$ function to allow ``forced'' backtracking if the length of the constructed word does not increase for a long time. Even with such a gadget, the time to build one walk of length $10^5$ varied from 9 minutes to 4 hours.
\begin{table}[!htb]
$$    \begin{array}{|c|c|c|c|c|}
    \hline
\text{Language}&\text{Mean } \Bf(\vec r)\ (10^5)&
\text{Mean } \Bf(\vec r)\ (5{\cdot}10^5)&2^{\mu(L)}& \gr(L)\\
\hline
\ASF&0.20475&0.20337&\approx1.15138&<1.44435;\ >1.00002\\
\ACF&1.08439&1.08418&\approx2.12017&< 2.37124;\ >1.02930\\
\A4F&0.20736&-&\approx1.15457&< 1.37417;\ >1.04427\\
\hline
    \end{array} $$
    \caption{Markov entropy for Abelian-power-free languages: experiments}
    \label{tab:apf}
\end{table}

\section{Conclusion and Future Work}

In this paper we showed that efficient sampling of very long random words is a useful tool in the study of factorial languages. Already the first experiments allowed us to state a lot of problems for further research. To mention just a few:\\
- for which classes of languages, apart from regular ones, the Markov entropy can be computed (or approximated with a given error)?\\
- are there natural classes of languages satisfying $\mu(L)=H(L)$? $\mu(L)\ll H(L)$?\\
- how the branching frequencies of walks in a prefix tree are distributed? which statistical tests can help to approximate this distribution?

Concerning the last questions, we note that though our experiments showed ``uniformity'' of branching frequencies in each of the studied languages, the frequencies of individual words can vary significantly. For example, the language $\PF(2,3)$ with the average frequency about $0.525$ contains  infinite words $\mathbf{u}$ and $\mathbf{v}$ satisfying $\Bf(\mathbf{u})=0.72$ and $\Bf(\mathbf{v})<0.45$ \cite{PeSh21}.

\bibliographystyle{splncs04}
\bibliography{bibl}

\end{document}